\documentclass{svjour3}

\usepackage{graphicx,epstopdf,multirow,multicol,subfigure,booktabs,pgfplots}
\let\proof\relax
\let\endproof\relax
\usepackage{amsmath}
\usepackage{amssymb,amsthm}
\usepackage{algorithm,algorithmic}
\usepackage[T1]{fontenc}
\usepackage{booktabs}
\usepackage{url}
\graphicspath{{./figures/}}
\allowdisplaybreaks
\smartqed

\begin{document}
	
\title{An Exact Algorithm for a Fuel-Constrained Autonomous Vehicle Path Planning Problem}

\author{Kaarthik Sundar \and Saravanan Venkatachalam \and Sivakumar Rathinam}


\institute{K. Sundar \at
Post-doctoral Researcher\\
Center for Non-Linear Studies, Los Alamos National Laboratory, NM 87544.\\
\email{kaarthik01sundar@gmail.com}           
\and
S. Venkatachalam \at
Assistant Professor \\
Dept. of Industrial and Systems Engg., Wayne State University, Detroit, MI 48202. 
\and
S. Rathinam \at
Associate Professor \\
Dept. of Mechanical Engg., Texas A\&M University, College Station, TX 77843.\\
}

\date{Received: date / Accepted: date}

\maketitle

\begin{abstract}
This paper addresses a fuel-constrained, autonomous vehicle path planning problem in the presence of multiple refueling stations. We are given a set of targets, a set of refueling stations, and a depot where $m$ vehicles are stationed. The vehicles are allowed to refuel at any refueling station, and the objective of the problem is to determine a route for each vehicle starting and terminating at the depot, such that each target is visited by at least one vehicle, the vehicles never run out of fuel while traversing their routes, and the total travel cost of all the routes is a minimum. We present four new mixed-integer linear programming formulations for the problem. These formulations are compared both analytically and empirically, and a branch-and-cut algorithm is developed to compute an optimal solution. Extensive computational results on a large class of test instances that corroborate the effectiveness of the algorithm are also presented. 

\end{abstract}

\keywords{fuel constraints \and green vehicle routing \and electric vehicles \and mixed-integer linear programming \and branch-and-cut
}

\section{Introduction}
\label{sec:intro}
Increasing concerns about climate change and rising green house gas emissions drive the research in sustainable and energy efficient mobility. One such example is the research involving a possible introduction of self driving, electrically-powered vehicles in the future. One of the main operational challenges for autonomous electric vehicles in transport applications is their limited range and the availability of recharging stations  \cite{Schneider2014,Hiermann2014}. As of November 2015, the number of electric stations in the US is a mere 9,571 with a total of 24,631 charging outlets \cite{USDOE}. Fig. \ref{fig:map} shows a map with the locations of electric recharging stations in Texas, USA; observe that the distribution of electric stations is very sparse except in the four major cities of Dallas, Houston, Austin, and San Antonio. Successful adoption of autonomous electric vehicles will strongly depend on methods to alleviate the range and recharging limitations. 

\begin{figure}
\centering
\includegraphics[scale=0.28]{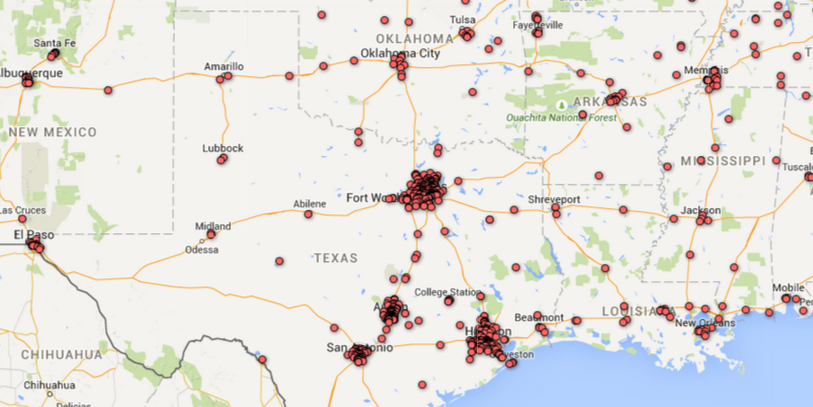}
\caption{Electric station locations in Texas, USA \cite{USDOE}}
\label{fig:map}
\end{figure}

Another example of sustainable and energy efficient mobility deals with vehicle routing applications involving green vehicles. Green vehicle routing problem is a variant of the Vehicle Routing Problem (VRP) and was introduced by authors in \cite{Erdougan2012} to account for the challenges associated with operating a fleet of alternate-fuel vehicles (AFVs). The US transportation sector accounts for 28\% of national greenhouse gas emissions \cite{USEPA}. Several efforts over many decades focusing towards the introduction of cleaner fuels (e.g. ultra low sulfur diesel) and efficient engine technologies have lead to reduced emissions and greater mileage per gallon of fuel used. Many government organizations, municipalities, and private companies are converting their fleet of vehicles to AFVs either voluntarily to alleviate the environmental impact of fossil based fuels or to meet environmental regulations. For instance, FedEx, in its overseas operations, employs AFVs that run on bio-diesel, liquid natural gas or compressed natural gas. 

In both the aforementioned applications, the following \textbf{Fuel-Constrained Autonomous Vehicle Path Planning Problem} (FCAVPP) naturally arises in order to efficiently manage a team of autonomous vehicles: Given a set of targets, a set of refueling stations, and a set of vehicles (AFVs or electric vehicles) stationed at a depot, find a route for each vehicle starting and ending at the depot such that (i) each target is visited at least once by some vehicle, (ii) no vehicle runs out of fuel as it traverses its path, and (iii) the total cost of the routes for the vehicles is minimized. We assume that whenever a vehicle visits a refueling station, it refuels to its capacity. 

In this paper, we formulate the FCAVPP as a combinatorial optimization problem. In particular, we introduce and compare four different mixed-integer linear programming (MILP) formulations for the FCAVPP. We use a well-known idea in combinatorial optimization called ``lifting'' \cite{Gomory1969} to compare the formulations analytically. The first two formulations are arc-based, and the rest are node-based formulations that use the Miller-Tucker-Zemlin (MTZ) constraints \cite{MTZ1960}. The arc-based and edge-based formulations have additional decision variables for each edge and vertex, respectively, to impose the fuel constraints on the vehicles. We then develop a branch-and-cut algorithm to compute an optimal solution to the FCAVPP based on the formulations. The importance of developing such formulations is two-fold: tight formulations provide lower bounds to the optimal objective value in the case of a minimization problem and this can in turn lead to an efficient way of benchmarking various heuristics to compute feasible solutions to the problem and secondly, heuristic solutions that are close to the optimal value when provided to such tight formulations can always lead to faster convergence of the branch-and-cut algorithm to the optimal solution of the original problem in very less computation time, typically within a few minutes (see \cite{Sundar2014}). In summary, the major contributions of this paper are as follows: (1) present four new MILP formulations for the FCAVPP, (2) compare the formulations both analytically and empirically, and (3) show through extensive computational experiments that instances with maximum of 40 targets are within the computational reach of a branch-and-cut algorithm based on the best of the four formulations.

\section{Related work}
\label{sec:litreview}
The FCAVPP is NP-hard because it is a generalization of the Traveling Salesman Problem (TSP). The existing literature on the FCAVPP is quite scarce. The single vehicle variant of the FCAVPP was first introduced by authors in \cite{Khuller2007}. When the travel costs are symmetric and satisfy the triangle inequality, authors in \cite{Khuller2007} provide an approximation algorithm for this variant. They assume that the minimum fuel required to travel from any target to its nearest depot is at most equal to $F\alpha/2$ units, where $\alpha$ is a constant in the interval $[0,1)$ and $F$ is the fuel capacity of the vehicle. This is a reasonable assumption as, in any case, one cannot have a feasible tour if there is a target that cannot be visited from any of the depots. Using these assumptions, Khuller et al. \cite{Khuller2007} present a $(3(1+\alpha))/(2(1-\alpha))$ approximation algorithm for the problem. Authors in \cite{Sundar2012,Sundar2014,Levy2014,Mitchell2015,Sundar2015} address variants of the FCAVPP in the context of path-planning for unmanned aerial vehicles. In particular, authors in \cite{Sundar2012} formulate the single vehicle variant as a MILP and present $k$-opt based exchange heuristics to obtain feasible solutions within $7\%$ of the optimal, on an average. They also provide the MILP formulation to off-the-shelf commercial solvers to obtain an optimal solution to the single vehicle variant and observe that the solvers are not able to handle their formulation for instances with more than 25 targets. Later, Sundar et al. \cite{Sundar2014} extend the approximation algorithm in \cite{Khuller2007} to the asymmetric case and also present heuristics to solve the asymmetric version of this variant. Furthermore, variable neighborhood search heuristics for the FCAVPP with heterogeneous vehicles, \emph{i.e.,} vehicles with different fuel capacities, are presented by Levy et al. \cite{Levy2014}. More recently, an approximation algorithm and heuristics are developed for FCAVPP in \cite{Mitchell2015}. The authors in \cite{Mitchell2015} extend the MILP proposed in \cite{Sundar2012} to the multiple vehicle setting. Similar to \cite{Sundar2012}, they also observed that the off-the-shelf mixed-integer solvers had difficulty computing an optimal solution for reasonable size test instances. To the best of our knowledge, there is no work in the literature that focuses on developing efficient formulations to obtain an optimal solution to even the single vehicle variant of the FCAVPP, let alone the FCAVPP.

Variants of the classic VRP that are closely related to FCAVPP include the distance constrained VRP \cite{Laporte1984,Li1992,Kara2010,Kara2011,Nagarajan2012}, the orienteering problem \cite{Fischetti1998,Vansteenwegen2011}, and the capacitated version of the arc routing problem \cite{Ghiani2004,Polacek2008}. The distance constrained VRP is a special case of FCAVPP with a single vehicle and single depot that can be considered as a fuel station. The FCAVPP is also quite different and more general compared to orienteering problem where one is interested in maximizing the number of targets visited by the vehicle subject to its fuel constraints. Lastly, the arc routing problem is a single depot VRP given a set of intermediate facilities, and the vehicle has to cover a subset of edges along which targets are present. The vehicle is required to collect goods from the targets as it traverses the given set of edges and unloads the goods at the intermediate facilities. The goal of this problem is to find a tour of minimum length that starts and ends at the depot such that the vehicle visits the given subset of edges, and the total amount of goods carried by the vehicle does not exceed the capacity of the vehicle along the tour. One of the key differences between the arc routing problem and FCAVPP is that there is no requirement that any subset of edges must be visited in FCAVPP.

There is also a class of online path planning algorithms in the literature that is very different from the FCAVPP; the difference is elaborated in this paragraph. Path planning for unmanned vehicles, aerial or ground, is typically performed at two levels. The FCAVPP and other variants of the VRP and TSP discussed in the previous paragraph are higher-level path planning algorithms. These problems are also refered to as routing problems \cite{Toth2001,Sundar2016,Sundar2016a} in the literature. These problems are of a combinatorial nature and are, in general, NP-Hard. Hence, these problems are solved offline, typically an hour before the actual mission. Furthermore, other issues like obstacle avoidance, closed-loop control, etc. are not taken into account at this level owing to the difficulty of solving these problems without these constraints. The lower level path planning and control algorithms then solve the problem of the trajectory generation for the vehicle to travel from one target to another in the presence of obstacles and compute the closed-loop control signals to follow the trajectory with minimum error; these algorithms are more real-time or online algorithms. Potential field methods, $A^*$ search algorithm, and Rapidly-exploring Random Tree (RRT) \cite{Latombe2012} are a few of the commonly used lower-level online path planning algorithms. The readers are refered to \cite{Yu2015} for a recent review of other state-of-the-art lower level path planning algorithms.

The remainder of the paper is organized as follows. Sec. \ref{sec:definition} states the formal definition of the problem and introduces the notations. Sec. \ref{sec:lifting} introduces the mathematical preliminaries on ``lifting for MILPs'' that is used in parts of this paper for strengthening a given MILP formulation. In Sec. \ref{sec:formulation}, we develop the four MILP formulations. The first two formulations are arc-based and the rest are node-based formulations \emph{i.e.,} decision variables for enforcing the fuel constraints are introduced for each edge and each target for the arc-based and the node-based formulations, respectively. The linear programming relaxations of the formulations are analytically compared in this section. Then, in Sec. \ref{sec:bandc} we detail the main steps involved in developing a branch-and-cut algorithm to solve the four formulations. Finally, in Sec. \ref{sec:results}, we present the computational results followed by conclusions and possible extensions.

\section{Problem definition \label{sec:definition}}
Let $T$ denote the set of targets $\{t_1,\dots,t_n\}$ and let $d_0$ denote the depot. Let $D$ denote the set of refueling stations $\{d_0, d_1,\dots, d_k\}$, and $d_0$ is the depot where $m$ vehicles are initially stationed with each vehicle fueled to its capacity. We assume that the vehicles cannot refuel at the depot. The FCAVPP is defined on a directed graph $G=(V,E)$ where $V=T\cup D$, and $E$ is the set of edges joining any two vertices in $V$; $E$ does not contain any self-loops. We assume that $G$ does not contain any self-loops. Each edge $(i,j) \in E$ is associated with a non-negative cost $c_{ij}$ required to travel from vertex $i$ to vertex $j$, and $f_{ij}$ is the fuel spent by traveling from $i$ to $j$. It is assumed that the cost of traveling from vertex $i$ to vertex $j$ is directly proportional to the fuel spent in traversing the edge $(i,j)$ \emph{i.e.}, $c_{ij} = K\cdot f_{ij}$ ($c_{ij}$ and $c_{ji}$ may be different, but for the purpose of this paper, we assume $c_{ij} = c_{ji}$). It is also assumed that travel costs satisfy the triangle inequality \emph{i.e.}, for every $i,j,k\in V$, $c_{ij} + c_{jk} \geq c_{ik}$. Furthermore, let $F$ denote the fuel capacity of all the vehicles. The FCAVPP consists of finding a route for each vehicle such that each vehicle $v_i$ starts and ends its route at $d_0$, each target is visited at least once by some vehicle, the vehicles never run out of fuel while they traverse their respective routes, and the sum of the cost of all the edges present in the routes is a minimum. Fig. \ref{fig:illustration} shows a feasible solution to the FCAVPP for an instance with four vehicles and five refueling stations. In the next section, we will briefly present a few mathematical preliminaries regarding ``lifting for MILPs'' that will be used in Sec. \ref{subsec:nodebased} to develop tighter formulations.
\begin{figure}
\centering
\includegraphics[scale=0.8]{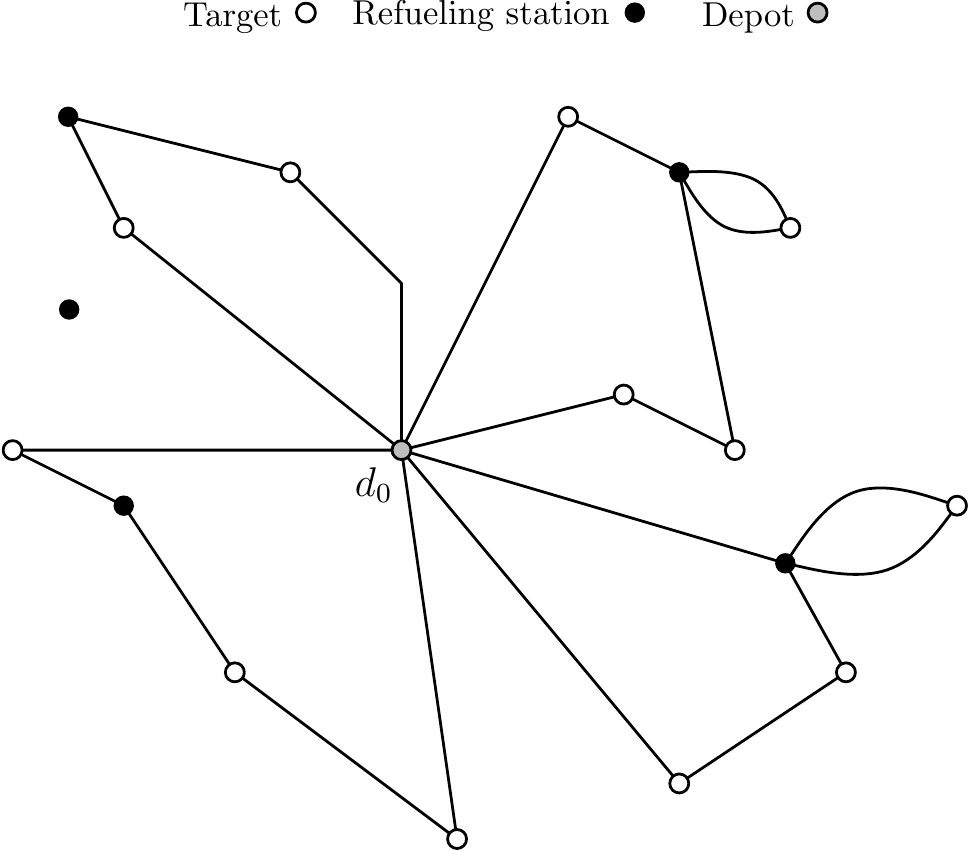}
\caption{A feasible solution to the FCAVPP. The solution is for a 4-vehicle problem. $d_0$ denotes the depot where all the vehicles are stationed initially. Note that a refueling station can be visited any number of times for refueling while some refueling stations may not be visited at all.}
\label{fig:illustration}
\end{figure}

\section{Lifting for MILPs}
\label{sec:lifting}
Before we present the idea of lifting for MILPs, we will review two useful definitions. A polyhedron is defined as the solution set of a finite number of linear equalities and inequalities, and a polytope is defined as a bounded polyhedron. A valid inequality for a polyhedron is an inequality that does not remove any feasible points of the polyhedron. Now, lifting is the process of constructing a valid inequality for a high dimensional polyhedron using a valid inequality for a low dimensional polyhedron. The idea of lifting was introduced by Gomory in 1969 \cite{Gomory1969}. The computational aspects of lifting were investigated by Padberg  \cite{Padberg1973}. Since then, lifting has been studied extensively \cite{Gu2000,Atamturk2004,Balas1978,Wolsey1976} and has been used for a variety of problems including vehicle routing, knapsack problem, supply chain, orienteering problems, to name a few. Lifting is usually applied sequentially; variables in a set are lifted one after another, and a separate optimization problem is solved to determine each lifting coefficient. The resulting inequality depends on the order in which the variables are lifted. 
\begin{proposition} 
A lifted inequality is guaranteed to have a dimension at least one greater than the original inequality, and the higher dimensional inequality is stronger than its corresponding lower dimensional counterpart.
\end{proposition}
\begin{proof}
See Ch. II.2, Prop. 1.1 and 1.2 in \cite{Wolsey2014}.
\end{proof}

Now, we illustrate the process of lifting using a numerical example. We will formulate an optimization problem to compute the lifting coefficient and observe that the lifted inequality is stronger than the original inequality. Consider the polytope $P$ defined by
\begin{flalign}
P := \{(x_1,x_2)\in \mathbb Z_+^2: x_1 + 3x_2 \leq 3, 3x_1 + x_2 \leq 3\} \label{eq:P}
\end{flalign}
where, $\mathbb Z_+^2$ is the set of non-negative integer points. $P$ contains the points $(0,0)$, $(1,0)$, and $(0,1)$. The inequality $x_1 \leq 1$ is a valid inequality for the polytope $P$ \emph{i.e.,} it does not remove any points in $P$. We will now lift this valid inequality to obtain a stronger valid inequality. The variable to be lifted is $x_2$, and we are interested in finding the best value of $\alpha$ (the lifting coefficient) such that the inequality $x_1 + \alpha x_2 \leq 1$ does not remove any feasible points in $P$. This can be formulated as the following optimization problem:
\begin{flalign}
\alpha = \max \{1- x_1 : (x_1,x_2) \in P, x_2 = 1\}.
\end{flalign}

Solving the above optimization problem, we obtain $\alpha = 1$, which is the best value for the lifting coefficient. The resulting inequality $x_1 + x_2 \leq 1$, a lifted version of $x_1 \leq 1$, is a better inequality for the polytope $P$. This simple procedure is used throughout the rest of the paper to strengthen various constraints for the FCAVPP and obtain tighter formulations.

\section{Mathematical formulations \label{sec:formulation}}
This section presents four formulations for the FCAVPP. The first two formulations are arc-based, and the remaining formulations are node-based. The arc-based and node-based formulations have additional decision variables for each edge and vertex, respectively, to impose the fuel constraints. For any given formulation $\mathcal F$, let $\mathcal F^L$ denote its linear programming relaxation obtained by allowing the integer variables to take continuous values within the lower and upper integer bounds, and $\operatorname{opt}(\mathcal F)$ denote the cost of its optimal solution. Additional notations that will be used in the formulation are as follows: for any set $S \subset V$, $\delta^+(S)=\{(i,j) \in E: i\in S, j\notin S\}$ and for any $A \subseteq E$, $x(A) = \sum_{(i,j)\in A} x_{ij}$.

\subsection{Arc-based formulations \label{subsec:arcbased}}
We first present an arc-based formulation $\mathcal F_1$ for the FCAVPP, inspired by the models for standard routing problems in \cite{Toth2001,Kara2011}. The decision variables used in the arc-based formulations are as follows: each edge $(i,j)\in E$ is associated with a variable $x_{ij}$, which equals $1$ if the edge $(i,j)$ is traversed by some vehicle, and $0$ otherwise. Also, associated with each edge $(i,j)$ is a flow variable $z_{ij}$ denoting the total fuel consumed by any vehicle as it starts from a depot and reaches the vertex $j$, when the predecessor of $j$ is $i$. Finally, associated with each refueling station $d \in D\setminus\{d_0\}$ is a binary variable $y_d$ that takes a value $1$ if it is used at least once by some vehicle and $0$ otherwise. Using the above variables, the first arc-based formulation $\mathcal F_1$ is given as follows:
\begin{flalign}
&(\mathcal F_1) \quad \text{Minimize} \quad \sum_{(i,j)\in E} c_{ij} x_{ij} \notag& \\
&\text{subject to:} \notag & \\
&\sum_{i\in V} x_{di} = \sum_{i\in V} x_{id} \quad \forall \, d\in D\setminus\{d_0\},\label{eq:f1-degree-d}& \\
&\sum_{i\in V} x_{di} \geq y_d \quad \forall \, d \in D\setminus\{d_0\}, \label{eq:f1-degree-d2} &\\
&\sum_{i\in V} x_{id_0} = m \text{ and } \sum_{i\in V} x_{d_0i} = m , \label{eq:f1-degree-t0} &\\
&\sum_{i\in V} x_{ij} = 1 \text{ and } \sum_{i\in V} x_{ji} = 1 \quad \forall \, j \in T, \label{eq:f1-degree-t} &\\
&x(\delta^+(S)) \geq y_d \quad \forall \, d\in S\cap D ,\, S\subset V\setminus\{d_0\} :  S\cap D \neq \emptyset,  \label{eq:f1-depotconnectivity} & \\
&\sum_{j\in V} z_{ij} - \sum_{j\in V} z_{ji} = \sum_{j\in V} f_{ij} x_{ij} \quad \forall \, i\in T, \label{eq:f1-fuel} &\\
&0 \leq z_{ij} \leq Fx_{ij} \quad \forall \, (i,j) \in E, \label{eq:f1-fuel2} &\\
&z_{di} = f_{di}x_{di} \quad \forall \, i\in T, \, d \in D \label{eq:f1-fuel3},   &\\
&x_{ij} \in \{0,1\} \quad \forall \, (i,j) \in E, \label{eq:f1-integer}  \text{ and}& \\
&y_d \in \{0,1\} \quad \forall \, d\in D\setminus \{d_0\}. \label{eq:f1-depoty}&
\end{flalign}
In the above formulation, the constraints \eqref{eq:f1-degree-d} and \eqref{eq:f1-degree-d2} state that the in-degree of each refueling station is equal to its out-degree. Constraints \eqref{eq:f1-degree-d2} also force $y_d$ for the refueling station $d \in D\setminus\{d_0\}$ to be $1$ if it is used by any vehicle. The constraints \eqref{eq:f1-degree-t0} ensure that all the vehicles leave and return to the depot. Constraints \eqref{eq:f1-degree-t} state that the in-degree and out-degree of each target must take a value $1$. Constraints \eqref{eq:f1-depotconnectivity} ensure that a feasible solution to  FCAVPP is connected. The constraints \eqref{eq:f1-fuel} eliminate sub-tours of the targets and also defines the flow variables $z_{ij}$ for each edge $(i,j) \in E$. The constraints \eqref{eq:f1-fuel2} and \eqref{eq:f1-fuel3} together impose $0\leq z_{ij} \leq F$ and they ensure that the fuel consumed by any vehicle to travel up to a depot does not exceed the fuel capacity of the vehicle, $F$. Finally, the constraints \eqref{eq:f1-integer} and \eqref{eq:f1-depoty} impose the binary restrictions on the decision variables $x_{ij}$ and $y_d$.

\begin{proposition} \label{prop:yint-relax}
The constraints \eqref{eq:f1-depoty} can be replaced by 
\begin{flalign}
& x_{di} \leq y_d \quad \forall\, i\in T\cup\{d_0\}, \, d\in D\setminus\{d_0\}, \text{ and } \label{eq:f1-fl} & \\
& 0\leq y_d \leq 1 \quad \forall \, d\in D\setminus \{d_0\}. \label{eq:f1-yrelax} & 
\end{flalign}
\end{proposition}
\begin{proof}
We first note that the constraints \eqref{eq:f1-fl} are valid inequalities for the FCAVPP. They state that the vehicles can use a refueling station $d$ only if $y_d = 1$. This constraint also ensures that the value of $y_d$ will either be $0$ or $1$ as $x_{di} \in \{0,1\}$ for any $i\in T \cup\{d_0\}$. Hence, the binary restrictions on the variables $y_d$ are relaxed in \eqref{eq:f1-yrelax}.
\end{proof}

Next, we present another arc-based formulation $\mathcal F_2$ which is a strengthened version of $\mathcal F_1$. The following proposition is a modified version of the Prop. 1 presented in \cite{Kara2011} for the distance-constrained vehicle routing problem; it strengthens the bounds given by the constraints \eqref{eq:f1-fuel2}.
\begin{proposition} \label{prop:strengthen} The inequalities in \eqref{eq:f1-fuel2} can be strengthened as follows:
\begin{flalign}
& z_{ij} \leq (F - t_j)x_{ij} \quad \forall j\in T,\, (i,j) \in E, \label{eq:f2-fuel1} & \\
& z_{id} \leq Fx_{id} \quad \forall i \in V \text{ and } d \in D, \label{eq:f2-fuel2} & \\
& z_{ij} \geq (s_i + f_{ij}) x_{ij} \quad \forall i\in T, \, (i,j) \in E \label{eq:f2-fuel3}, &
\end{flalign}
where, $t_i = \min_{d\in D} f_{id}$ and $s_i = \min_{d\in D} f_{di}$.
\end{proposition}
\proof
When $j \in D$ or $i \in D$, the constraints \eqref{eq:f2-fuel2} and \eqref{eq:f1-fuel3} specify the values of $z_{id}$ and $z_{di}$, respectively. When both $i,j \in D$, the constraint \eqref{eq:f1-fuel2} bounds the value of $z_{ij}$. Hence, we need only to discuss the case when $i,j \in T$. When $x_{ij} = 1$, the total fuel consumed by any vehicle that traverses the edge $(i,j)$ cannot be greater that $(F - t_j)$, where $t_j$ is the minimum amount of fuel required by any vehicle to reach a refueling station or the depot from target $j$. Therefore, the constraint in \eqref{eq:f2-fuel1} strengthens the upper bound of $z_{ij}$ in \eqref{eq:f1-fuel2}. Similarly, any vehicle that traverses the edge $(i,j)$ consumes at least $(s_i + f_{ij})$ amount of fuel. As a result, the constraint in \eqref{eq:f2-fuel3} strengthens the lower bound of $z_{ij}$ in \eqref{eq:f1-fuel2}.
\endproof
Hence, the second arc-based formulation is as follows:
\begin{flalign}
&(\mathcal F_2) \quad \text{Minimize} \quad \sum_{(i,j)\in E} c_{ij} x_{ij} \notag& \\
&\text{subject to: \eqref{eq:f1-degree-d} -- \eqref{eq:f1-fuel}, \eqref{eq:f1-fuel3}, \eqref{eq:f1-integer}, and \eqref{eq:f1-fl} -- \eqref{eq:f2-fuel3}.} \notag &
\end{flalign}
\begin{corollary} \label{cor:LP} $\operatorname{opt}(\mathcal F_2^L) \geq \operatorname{opt}(\mathcal F_1^L)$. \hfill \qed
\end{corollary}

\subsection{Node-based formulations \label{subsec:nodebased}}
In this section, we present a node-based formulation for the FCAVPP based on the models for the distance-constrained VRP in \cite{Desrochers1991,Kara2010}. For the node-based formulation, apart from the binary variable $x_{ij}$ for each edge $(i,j) \in E$, we have an auxiliary variable $u_i$ for each target $i$ to eliminate sub-tours of the targets disconnected from the vertices in set the $D$ and to enforce the fuel constraints. In addition, we will also use the following two parameters: $t_i = \min_{d\in D} f_{id}$ and $s_i = \min_{d\in D} f_{di}$ for every vertex $i \in V$. For any $d \in D$, $t_d = 0$ and $s_d = 0$. Using the above notations, the formulation $\mathcal F_3$ is given as follows:
\begin{flalign}
&(\mathcal F_3) \quad \text{Minimize} \quad \sum_{(i,j)\in E} c_{ij} x_{ij} \notag& \\
&\text{subject to: \eqref{eq:f1-degree-d} -- \eqref{eq:f1-depotconnectivity}, \eqref{eq:f1-integer}, \eqref{eq:f1-fl} -- \eqref{eq:f1-yrelax}}, \notag & \\
& u_i - u_j + M_{ij} x_{ij} \leq M_{ij} - f_{ij} \quad \forall i, j \in T, \label{eq:f3-fuel1}& \\
& u_i \geq s_i + \sum_{d\in D} (f_{di} - s_i)x_{di} \quad \forall i \in T \label{eq:f3-fuel2},  \text{ and} & \\
& u_i \leq F - t_i - \sum_{d \in D} (f_{id}-t_i)x_{id} \quad \forall i \in T. \label{eq:f3-fuel3} &
\end{flalign}
The constraints \eqref{eq:f3-fuel1} serve two purposes: they eliminate sub-tours of targets that are disconnected from any vertex in the set $D$ and any path connecting two vertices in the set $D$ that consume greater than $F$ amount of fuel (see Prop. \ref{prop:validity}). These constraints are a modified version of the MTZ constraints. The value of $M_{ij}$ in \eqref{eq:f3-fuel1} is given by $M_{ij} = F - s_j -t_i + f_{ij}$. The constraints \eqref{eq:f3-fuel2} and \eqref{eq:f3-fuel3} specify the upper and lower bounds on $u_i$, for every target $i$. The next proposition proves the validity of the constraint \eqref{eq:f3-fuel1} for the FCAVPP; it is similar to a proposition presented for the distance-constrained VRP \cite{Desrochers1991}.

\begin{proposition} \label{prop:validity}
The constraint \eqref{eq:f3-fuel1} eliminates sub-tours of the targets disconnected from any vertex in the $D$ and any path connecting two vertices in $D$ that requires more than $F$ amount of fuel for any vehicle.
\end{proposition}
\proof 
Suppose $(t_1, t_2, \dots, t_k, t_{k+1} \equiv t_1)$ is a sub-tour of the targets disconnected from any vertex in the set $D$. Aggregating the constraints \eqref{eq:f3-fuel1} corresponding to each edge in the sub-tour, we obtain 
\begin{flalign*}
&\sum_{i=1}^k (u_{t_i} - u_{t_{i+1}} + f_{t_i t_{i+1}}) \leq 0 \Rightarrow \sum_{i=1}^k f_{t_i t_{i+1}} \leq 0, &
\end{flalign*}
a contradiction. Hence, the constraint \eqref{eq:f3-fuel1} eliminates sub-tours of the targets disconnected from any vertex in the set $D$. Now, let $(v_1, v_2, \dots, v_{k-1}, v_k)$ denote a path such that $v_1, v_k \in D$ and $v_i \in T$ for $i\in \{2,\dots,k-1\}$. Again, aggregating the constraints \eqref{eq:f3-fuel1} for the edges $(v_j,v_{j+1})$, $j=2,\dots, k-2$, we obtain 
\begin{flalign*}
&u_{v_2} - u_{v_{k-1}} + \sum_{j=2}^{k-2} f_{v_j v_{j+1}} \leq 0. & 
\end{flalign*}
But, by constraints \eqref{eq:f3-fuel2} and \eqref{eq:f3-fuel3}, we have $u_{v_2} \geq f_{v_1 v_2}$ and  $u_{v_{k-1}} \leq F - f_{v_{k-1} v_k}$, respectively. Combining all the three inequalities, we obtain 
\begin{flalign*}
& \sum_{i=1}^{k-1} f_{v_i, v_{i+1}} \leq F. &
\end{flalign*}
Hence, any vehicle that traverses a path starting and terminating at vertices in $D$, consumes less than $F$ amount of fuel.
\endproof

The following proposition strengthens the constraints \eqref{eq:f3-fuel1}--\eqref{eq:f3-fuel3}. It uses the idea of lifting (see Sec. \ref{sec:lifting}) to tighten these constraints.

\begin{proposition} \label{prop:lifting} The inequalities in \eqref{eq:f3-fuel1}, \eqref{eq:f3-fuel2}, and \eqref{eq:f3-fuel3} can be strengthened as follows:
\begin{flalign}
& u_i - u_j + M_{ij} x_{ij} + (M_{ij}-f_{ij}-f_{ji})x_{ji}  \leq  M_{ij} - f_{ij} \quad  \forall i,j \in T, \label{eq:f3-fuel4} \\
& u_i \geq \sum_{j \in V} (s_j + f_{ji}) x_{ji} \quad \forall i \in T, \text{ and} \label{eq:f3-fuel5}  & \\
& u_i \leq F - \sum_{j \in V} (t_j + f_{ij}) x_{ij} \quad \forall i \in T \label{eq:f3-fuel6} & 
\end{flalign}
where, $x_{ii} = 0$ and $x_{ij} = 0$ whenever $s_i + f_{ij} + t_j > F$.
\end{proposition}
\proof
The constraints \eqref{eq:f3-fuel4} are obtained by lifting each variable $x_{ji},~j\in T$ in \eqref{eq:f3-fuel1}, in any order. Let $j$ be an arbitrary target in $T$. We compute its lifting coefficient as the maximum value of $\alpha$ such that the following constraint is valid to the FCAVPP: 
\begin{flalign}
&u_i - u_j + M_{ij} x_{ij} + \alpha x_{ji} \leq M_{ij} - f_{ij}. \label{eq:via-1}&
\end{flalign}
The above inequality is valid when $x_{ji} = 0$, as it reduces to \eqref{eq:f3-fuel1}. So, we suppose $x_{ji} = 1$; then, we have $x_{ij} = 0$ and $u_j + f_{ji} \leq u_i$. The constraint \eqref{eq:via-1} simplifies to
\begin{flalign*}
&\alpha \leq M_{ij} - f_{ij} + u_j - u_i \Rightarrow \alpha \leq M_{ij} - f_{ij} - f_{ji}. &
\end{flalign*}
Hence, the lifting coefficient for the variable $x_{ji}$ is given by $M_{ij} - f_{ij} - f_{ji}$.

Similarly, \eqref{eq:f3-fuel5} can be obtained by lifting each $x_{ji}$ variable for $j\in T$ in constraint \eqref{eq:f3-fuel2}, in any order. To that end, let $j\in T$. We compute the maximum value that $\alpha$ can take so that the following constraint is valid to the FCAVPP: 
\begin{flalign}
& u_i \geq s_i + \sum_{d\in D} (f_{di} - s_i)x_{di} + \alpha x_{ji}. \label{eq:via-2} &
\end{flalign}
The above inequality is valid for $x_{ji} = 0$. When $x_{ji} = 1$, we have $x_{di} = 0$ and $u_j + f_{ji} \leq u_i$. Hence, we have 
\begin{flalign*}
&\alpha \leq u_i - s_i \leq u_j + f_{ji} - s_i.&
\end{flalign*}
The minimum value of $u_j$ is equal to $s_j$. Hence, the lifting coefficient for $x_{ji}$ is equal to  $s_j + f_{ji} - s_i$. Similarly, the coefficients of the other $x_{ji}$ variables can be computed. Using these extra terms, the constraint \eqref{eq:via-2} can be simplified as follows:
\begin{flalign*}
&~u_i \geq s_i + \sum_{j \in V} (s_j + f_{ji}-s_i) x_{ji} \quad \forall i \in T & \\
\Rightarrow &~u_i \geq s_i + \sum_{j\in V} (s_j + f_{ji})x_{ji} - s_i \cdot \sum_{j\in V}x_{ji} \quad \forall i \in T& \\
\Rightarrow &~u_i \geq \sum_{j\in V} (s_j + f_{ji})x_{ji} \quad \forall i \in T.&
\end{flalign*}
The last implication follows from the constraints \eqref{eq:f1-degree-t}. 

Finally, the constraints \eqref{eq:f3-fuel6} are obtained from \eqref{eq:f3-fuel3} by lifting each $x_{ij}$ variable for every $j\in T$ in any order. Since, the procedure for computing the lifting coefficients is similar, we list the coefficients without proof. The lifting coefficient for the variable $x_{ij}$ where $j \in T$ is equal to $f_{ij} + t_j - t_i$. Further simplification using the degree constraint \eqref{eq:f1-degree-t} yields:
\begin{flalign*}
&u_i \leq F - \sum_{j \in V} (t_j + f_{ij}) x_{ij}.&&
\end{flalign*}
\endproof

Now, we will replace the constraints \eqref{eq:f3-fuel1}--\eqref{eq:f3-fuel3} in the node-based formulation $\mathcal F_3$ using its strengthened counterparts, \eqref{eq:f3-fuel4}--\eqref{eq:f3-fuel6}. Then, we observe that for the first target $r$ that any vehicle visits either from the depot or a refueling station $d \in D$, the constraints \eqref{eq:f3-fuel5} and \eqref{eq:f3-fuel6} imply 
\begin{flalign*}
&~f_{dr} \leq u_r \leq F - \sum_{j \in V} (f_{rj} + t_j) x_{rj} & \\
\Rightarrow &~f_{dr} \leq u_r \leq F - \max_{j\in V} (f_{rj} + t_j). & 
\end{flalign*}
This indicates that the value of each auxiliary variable $u_i,~i\in T$ in a solution to the FCAVPP obtained using the new formulation may not provide any meaningful information \cite{Kara2010}. This can be easily rectified by using the following valid inequality
\begin{flalign}
& u_i \leq F - t_i - \sum_{d\in D} (F-t_i - f_{di}) x_{di}. \label{eq:f3-fuel7} &
\end{flalign}
The above valid inequality tightens the upper bound on the auxiliary variable $u_i$, whenever $i$ is the first target visited by any vehicle as it leaves the depot or a refueling station. Hence, in lieu of the constraints \eqref{eq:f3-fuel4}--\eqref{eq:f3-fuel7}, the $u_i$s can be redefined as the total fuel consumed by any vehicle as it reaches the target $i$. The redefined $u_i$s are a more natural choice of decision variables to impose the fuel restrictions of the vehicles. These are the variables that authors in \cite{Sundar2012} and \cite{Mitchell2015} use to formulate the single vehicle and multiple vehicle fuel-constrained vehicle routing problem. But, we will observe in the forthcoming Sec. \ref{sec:results} that this choice is not well suited for obtaining an optimal solution to the FCAVPP.

Now, we present the second node-based formulation as follows:
\begin{flalign}
&(\mathcal F_4) \quad \text{Minimize} \quad \sum_{(i,j)\in E} c_{ij} x_{ij} \notag& \\
&\text{subject to: \eqref{eq:f1-degree-d} -- \eqref{eq:f1-depotconnectivity}, \eqref{eq:f1-integer}, \eqref{eq:f1-fl} -- \eqref{eq:f1-yrelax}, and \eqref{eq:f3-fuel4} -- \eqref{eq:f3-fuel7}} \notag. &
\end{flalign}
The following result follows from Prop. \ref{prop:lifting}.
\begin{corollary} \label{cor:LPnode} $\operatorname{opt}(\mathcal F_4^L) \geq \operatorname{opt}(\mathcal F_3^L)$. \hfill \qed
\end{corollary}

\section{Branch-and-cut algorithm} \label{sec:bandc}
In this section, we briefly present the main ingredients of a branch-and-cut algorithm that is used to solve the four different formulations presented in the previous section to optimality. The formulations developed in the Sec. \ref{sec:formulation} can be provided to off-the-shelf commercial branch-and-cut solvers to obtain an optimal solution to the FCAVPP. But, observe that all the formulations contain constraint \eqref{eq:f1-depotconnectivity} to ensure any feasible solution to FCAVPP is connected. The number of such constraints is exponential and it may not be computationally efficient to enumerate all these constraints and provide them to these solvers. To address this issue, we use the following approach: we relax the constraints \eqref{eq:f1-depotconnectivity} from the formulation, and whenever the solver obtains an integer feasible solution to this relaxed problem, we check if any of the constraints \eqref{eq:f1-depotconnectivity} are violated by the integer feasible solution. If so, we add the infeasible constraint and continue solving the original problem. This process of adding constraints to the problem sequentially has been observed to be computationally efficient for the TSP, VRP and a huge number of their variants \cite{Toth2001}. 

Now, we will detail the algorithm used to find a constraint  \eqref{eq:f1-depotconnectivity} that is violated for a given integer feasible solution to the relaxed problem. We will use formulation $\mathcal F_1$ to present the algorithm. The same algorithm can be used for the remaining formulations, without any modifications, to find violated constraints \eqref{eq:f1-depotconnectivity}. A violated constraint \eqref{eq:f1-depotconnectivity} can be described by a subset of vertices $S \subset V\setminus\{d_0\}$ such that $S\cap D \neq \emptyset$ and $x(\delta^+(S)) < y_d$ for every $d \in S\cap D$. Given an integer feasible solution to the relaxed problem, we construct a support graph $G^*$ defined as follows: $G^* = (V^*,E^*)$ where $V^*=T\cup \{d_0\} \cup \{d \in D\setminus\{d_0\}: y_d = 1\}$ and $E^* = \{(i,j) \in E: x_{ij} = 1\}$. We then find the strongly connected components of $G^*$. Every strongly connected component that does not contain the depot is a subset $S$ of $V\setminus\{d_0\}$ which violates the constraint \eqref{eq:f1-depotconnectivity}. We add all these infeasible constraints and continue solving the original problem. Many off-the-shelf commercial solvers \cite{cplex} provide a feature called ``solver callbacks'' to implement such an algorithm into its branch-and-cut framework.

\section{Computational results} \label{sec:results}
In this section, we discuss the computational performance of the branch-and-cut algorithm for all the four formulations presented in the Sec. \ref{sec:formulation}. The mixed-integer linear programs were implemented in Java, using the traditional branch-and-cut framework and the solver callback functionality of CPLEX version 12.6.2. All the simulations were performed on a Dell Precision T5500 workstation (Intel Xeon E5630 processor @2.53 GHz, 12 GB RAM). The computation times reported are expressed in seconds, and we imposed a time limit of 3,600 seconds for each run of the algorithm. The performance of the algorithm was tested with randomly generated test instances. \\

\noindent {\it Instance generation}

The problem instances were randomly generated in a square grid of size [100,100]. The number of refueling stations was set to 4 and the locations of the depot and all the refueling stations were fixed a priori for all the test instances. The number of targets varies from $10$ to $40$ in steps of five, while their locations were uniformly distributed in the square grid; for each $|T| \in \{10,15,20,25,30,25,40\}$, we generated five random instances. For each of the above generated instances, the number of vehicles in the depot was varied from $3$ to $5$, and the fuel capacity of the vehicles, $F$, was varied linearly with a parameter $\lambda$. $\lambda$ is defined as the maximum distance between the depot and any target. The fuel capacity $F$ was assigned a value from the set $\{2.25 \lambda, 2.5 \lambda, 2.75 \lambda, 3\lambda \}$. The travel costs and the fuel consumed to travel between any pair of vertices were assumed to be directly proportional to the Euclidean distances between the pair and rounded down to the nearest integer. In total, our test bed consisted of $420$ instances. Two sets of computational experiments were performed on these instances. The first experiment was aimed at comparing the linear programming relaxations of the four formulations and ensuring that the computational behaviour of the formulations conforms with the corollaries \ref{cor:LP} and \ref{cor:LPnode}. The second set of test runs was aimed at finding the most useful formulation to compute an optimal solution to any instance of the problem using the branch-and-cut algorithm proposed in Sec. \ref{sec:bandc}.  \\

Tables \ref{tab:1} and \ref{tab:2}, and Fig. \ref{fig:times}--\ref{fig:cuts} summarize the computational behavior of the algorithms. The following nomenclature is used throughout the rest of the paper:\medskip{}

\noindent $\#$: instance number;

\noindent $\operatorname{opt}(\mathcal F_i^L)$: linear programming relaxation solution for formulation $i$;

\noindent $n$: instance size \emph{i.e.}, number of targets in the instance;

\noindent $\#$ cuts: number of constraints \eqref{eq:f1-depotconnectivity} violated during a branch-and-cut solve. These violated constraints are computed using the algorithm detailed in Sec. \ref{sec:bandc};

\noindent total: total number of test instances of a given size (the number of targets and the number of vehicles are given);

\noindent succ: number of instances for which optimal solutions were computed within a time limit of 3,600 seconds.  \\

\noindent Table \ref{tab:1} compares the cost of the linear programming (LP) relaxations of the four formulations presented in Sec. \ref{sec:formulation} for the instances with 40 targets and 5 vehicles.  The results in Table \ref{tab:1} provide an empirical comparison of the formulations presented in \ref{sec:formulation}; the observed behavior is expected because the formulations $\mathcal F_2$ and $\mathcal F_4$ are strengthened versions of $\mathcal F_1$ and $\mathcal F_3$, respectively (see corollaries \ref{cor:LP} and \ref{cor:LPnode}). As for the LP relaxations of formulations $\mathcal F_2$ and $\mathcal F_4$, the formulation $\mathcal F_4$ is observed to have a better LP relaxation value for 95\% of the instances. Also, the formulation $\mathcal F_2$ is consistently better than the formulation $\mathcal F_3$ with respect to the strength of the LP relaxation. Hence, the rest of the computational experiments is devoted to comparing the formulations $\mathcal F_2$ and $\mathcal F_4$ with regards to computing an optimal solution to any instance of the FCAVPP.

\begin{table}
\centering
\begin{tabular}{lrrrr}
\toprule
$\#$ & $\operatorname{opt}(\mathcal F_1^L)$ & $\operatorname{opt}(\mathcal F_2^L)$ & $\operatorname{opt}(\mathcal F_3^L)$ & $\operatorname{opt}(\mathcal F_4^L)$ \\
\midrule
1 & 572.051 & 587.347 & 551.014 & 620.000\\
2 & 567.454 & 578.639 & 551.014 & 620.000 \\
3 & 564.260 & 572.561 & 551.014 & 620.000 \\
4 & 561.937 & 568.191 & 551.014 & 620.000 \\
5 & 561.001 & 574.636 & 535.011 & 590.500 \\
6 & 558.195 & 568.381 & 535.011 & 590.500 \\
7 & 556.121 & 563.924 & 535.011 & 590.500 \\
8 & 554.441 & 560.568 & 535.011 & 590.500 \\
9 & 493.041 & 497.721 & 477.507 & 544.000 \\
10 & 489.625 & 493.075 & 477.507 & 544.000 \\
11 & 486.908 & 489.610 & 477.507 & 544.000 \\
12 & 484.697 & 486.902 & 477.507 & 544.000 \\
13 & 604.135 & 617.703 & 574.005 & 611.500 \\
14 & 599.296 & 607.671 & 574.005 & 611.500 \\
15 & 595.086 & 600.576 & 574.005 & 611.500 \\
16 & 591.773 & 596.955 & 574.005 & 611.500 \\
17 & 593.720 & 607.825 & 563.007 & 616.000 \\
18 & 588.627 & 599.340 & 563.007 & 616.000 \\
19 & 584.692 & 592.348 & 563.007 & 616.000 \\
20 & 581.550 & 587.525 & 563.007 & 616.000 \\
\bottomrule
\end{tabular}
\caption{Cost of the LP relaxation for the instances with 40 targets and 5 vehicles.}
\label{tab:1}
\end{table}

Table \ref{tab:2} shows the number of instances of different sizes solved to optimality by the formulations $\mathcal F_2$ and $\mathcal F_4$ within the time limit of 3,600 seconds. For a fixed number of targets $n$, and number of vehicles $m$, the branch-and-cut algorithm for formulations $\mathcal F_2$ and $\mathcal F_4$ was run on a set of 20 instances. The column titled ``succ'' indicates the number of instances, out of 20, that were solved to optimality. We observe that the formulation $\mathcal F_2$ is more successful in computing an optimal solution for the instances with 35 and 40 targets. In total, the formulation $\mathcal F_2$ provided an optimal solution to 399 out of 420 test instances within 3,600 seconds. In contrast, the formulation $\mathcal F_4$ provided an optimal solution only to 331 test instances within the stipulated time. We also note that the formulation $\mathcal F_2$, for the 21 instances that were not solved to optimality, provided feasible solutions that were within 2.5\% of the optimal value, on an average. 

The Fig. \ref{fig:plot} shows an optimal solution that is obtained for a 25-target, 2-vehicle instance. The formulation $\mathcal F_2$ is used to compute the optimal solution and the computation time for the instance is 78.05 seconds.

\begin{figure}
\centering
\includegraphics[scale=0.5]{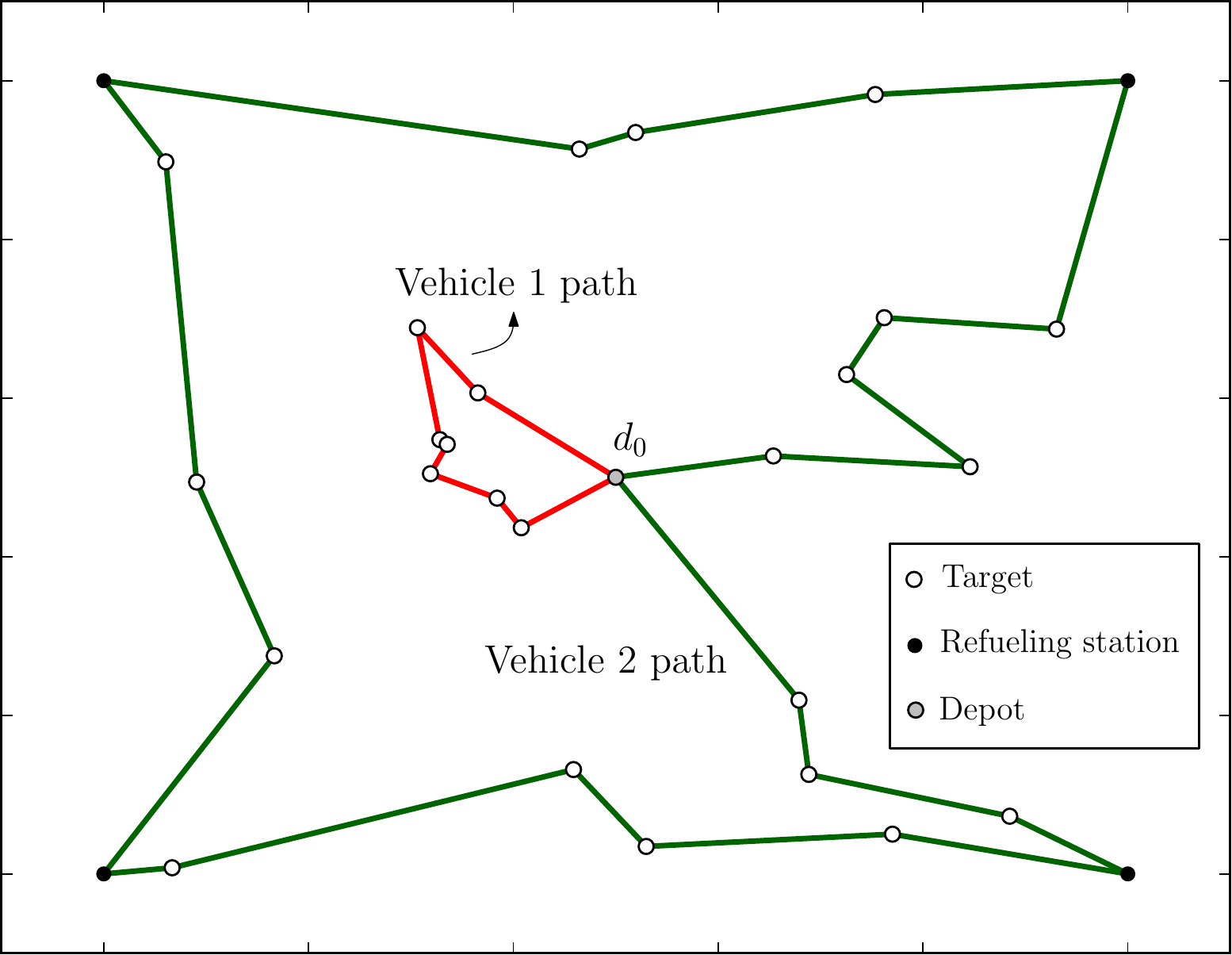}
\caption{Optimal solution for a 25-target, 2-vehicle instance.}
\label{fig:plot}
\end{figure}

Fig. \ref{fig:times} shows the average time taken by the two formulations to compute an optimal solution. The computation time reported in Fig. \ref{fig:times} is averaged over all the instances for a given value of $n$, the number of targets. Both table \ref{tab:2} and Fig. \ref{fig:times} indicate that the arc-based formulation $\mathcal F_2$ outperforms the node-based formulation $\mathcal F_4$ for around 95\% of the instances. The remaining 5\% of the instances on which the two formulations were indistinguishable contained less than 15 targets.

Fig. \ref{fig:times1} and Fig. \ref{fig:cuts} analyse the behaviour of the branch-and-cut algorithm using the arc-based formulation $\mathcal F_2$. Fig. \ref{fig:times1} compares the average time taken to compute an optimal solution to the FCAVPP as the number of vehicles is varied from 3 to 5. The plot indicates that the computation time decreases with the increase in number of vehicles. This is not surprising because the combinatorial difficulty of the problem arises due to the fuel constraints of the vehicles. Since all the vehicles are forced to leave the depot, the fuel constraints tend to be irrelevant with the increase in the number of vehicles. Finally, Fig. \ref{fig:cuts} shows the average number of connectivity constraints \eqref{eq:f1-depotconnectivity}, as the number of vehicles and the number of targets is varied from 3 to 5 and 20 to 40, respectively, that were actually violated and added back to the relaxed formulation by the branch-and-cut algorithm. The plot corroborates the proposed approach of relaxing the exponential number of connectivity constraints \eqref{eq:f1-depotconnectivity} and adding them back the formulation when a violation occurs. We observe that the number of violated constraints \eqref{eq:f1-depotconnectivity} is much less compared to the original number of constraints in the formulation. 

\begin{table}
\centering
\begin{tabular}{lrrrrrrr}
\toprule
 & & \multicolumn{2}{c}{$m=3$} & \multicolumn{2}{c}{$m=4$} & \multicolumn{2}{c}{$m=5$} \\
 \cline{3-8}
 \addlinespace[0.5em]
 & & $\mathcal F_2$ & $\mathcal F_4$ & $\mathcal F_2$ & $\mathcal F_4$ & $\mathcal F_2$ & $\mathcal F_4$ \\
 \addlinespace[0.5em]
 $n$ & total & succ & succ & succ & succ & succ & succ\\
\midrule
10 & 20 & 20 & 20 & 20 & 20 & 20 & 20\\
15 & 20 & 20 & 20 & 20 & 20 & 20 & 20\\
20 & 20 & 20 & 20 & 20 & 20 & 20 & 20\\
25 & 20 & 19 & 18 & 20 & 20 & 20 & 20\\
30 & 20 & 20 & 12 & 20 & 18 & 20 & 20\\
35 & 20 & 16 & 7  & 18 & 10 & 18 & 10\\
40 & 20 & 16 & 3  & 16 & 5  & 16 & 8 \\
\bottomrule
\end{tabular}
\caption{Comparison of formulations $\mathcal F_2$ and $\mathcal F_4$.}
\label{tab:2}
\end{table}

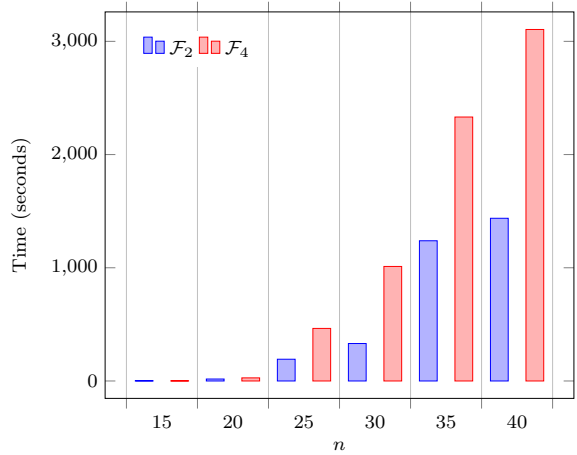
\begin{figure}
	\centering
\begin{tikzpicture}[scale=0.9]
\begin{axis}[
	x tick label style={
		/pgf/number format/1000 sep=},
	ylabel=Time (seconds),
	xlabel=$n$,
	enlargelimits=0.05,
	legend style={at={(0.2,0.95)}, draw=none,
	anchor=north,legend columns=-1},
	ybar interval=0.5,
]
\addplot
	coordinates {(15,1.92) (20,17.29) (25,192.12) (30,330.91) (35,1238.53) (40,1437.18) (45,709)};
\addplot
	coordinates {(15,1.39) (20,27.39) (25,464.30) (30,1011.65) (35,2331.55) (40,3104.88) (45,709)};
\legend{$\mathcal F_2$,$\mathcal F_4$}
\end{axis}
\end{tikzpicture}
\caption{Average time taken to compute an optimal solution for $\mathcal F_2$ and $\mathcal F_4$.}
\label{fig:times}
\end{figure}

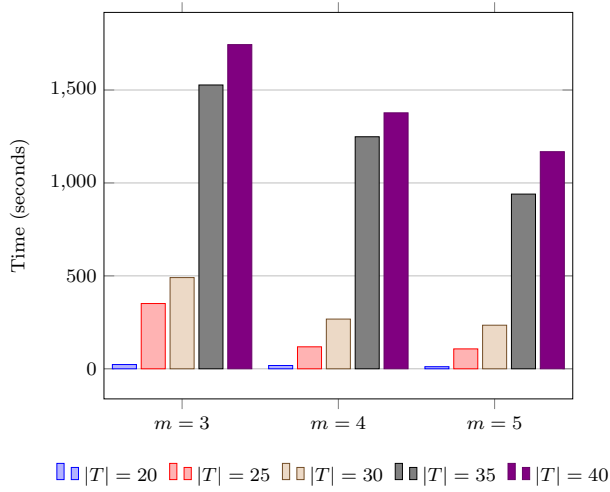
\begin{figure}
	\centering
\begin{tikzpicture}[scale=0.9]
\begin{axis}[
    ybar,
    enlarge x limits=0.25,
    legend style={at={(0.5,-0.15)},draw=none,
      anchor=north,legend columns=-1},
    ylabel={Time (seconds)},
    symbolic x coords={$m=3$,$m=4$,$m=5$},
    xtick=data,
    ymajorgrids,
    ]
\addplot coordinates {($m=3$,22.87) ($m=4$,17.67) ($m=5$,11.31)};
\addplot coordinates {($m=3$,351.12) ($m=4$,118.21) ($m=5$,107.02)};
\addplot coordinates {($m=3$,490.75) ($m=4$,267.48) ($m=5$,234.49)};
\addplot coordinates {($m=3$,1527.10) ($m=4$,1248.51) ($m=5$,939.98)};
\addplot coordinates {($m=3$,1744.27) ($m=4$,1377.23) ($m=5$,1168.16)};
\legend{$|T|=20~$,$|T|=25~$,$|T|=30~$,$|T|=35~$,$|T|=40~$}
\end{axis}
\end{tikzpicture}
\caption{Average time taken to compute the optimal solution for the formulation $\mathcal F_2$.}
\label{fig:times1}
\end{figure}

\begin{figure}
	\centering
\begin{tikzpicture}[scale=0.9]
\begin{axis}[
    ybar,
    enlarge x limits=0.25,
    legend style={at={(0.5,-0.15)},draw=none,
      anchor=north,legend columns=-1},
    ylabel={$\#$ cuts},
    symbolic x coords={$m=3$,$m=4$,$m=5$},
    xtick=data,
    ymajorgrids,
    ]
\addplot coordinates {($m=3$,9) ($m=4$,7) ($m=5$,6)};
\addplot coordinates {($m=3$,14) ($m=4$,11) ($m=5$,10)};
\addplot coordinates {($m=3$,14) ($m=4$,10) ($m=5$,10)};
\addplot coordinates {($m=3$,14) ($m=4$,12) ($m=5$,11)};
\addplot coordinates {($m=3$,15) ($m=4$,13) ($m=5$,10)};
\legend{$|T|=20~$,$|T|=25~$,$|T|=30~$,$|T|=35~$,$|T|=40~$}
\end{axis}
\end{tikzpicture}
\caption{Average number of violated constraints \eqref{eq:f1-depotconnectivity} identified by the branch-and-cut algorithm for the formulation $\mathcal F_2$.}
\label{fig:cuts}
\end{figure}
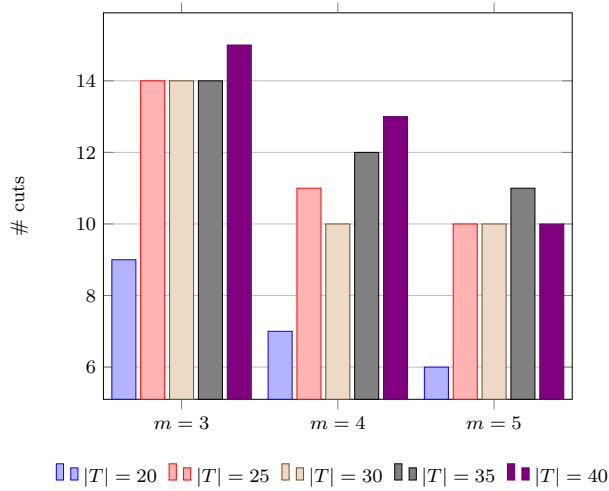

\section{Conclusions and future work \label{sec:conclusion}}
In this paper, we presented four different mixed-integer linear programming formulations for fuel-constrained autonomous vehicle path planning problem. The problem arises frequently in the context of green vehicle routing, routing electric vehicles, and path planning for UAVs. The first two formulations $\mathcal F_1$ and $\mathcal F_2$ are arc-based, and the rest \emph{i.e,} $\mathcal F_3$ and $\mathcal F_4$ are node-based formulations. The arc-based formulations use decision variables on the edges to impose the fuel constraints and the node-based formulations use extra decision variables on the targets to enforce the fuel constraints. Though the node-based formulations are a more natural way of formulating the FCAVPP, the arc-based formulations outperform the node-based formulations with respect to computing an optimal solution to any instance of the problem. Computational experiments on a large number of test instances corroborate this observation. Among the two proposed arc-based formulations, the second formulation $\mathcal F_2$ is shown to be the most effective, both analytically and empirically. Hence, $\mathcal F_2$ can be used as a template for enforcing fuel constraints on any problem that includes similar fuel restrictions on the vehicle. Future work can be directed towards developing similar mixed-integer linear programming formulations and branch-and-cut algorithms to solve a heterogeneous variant of the problem \emph{i.e.,} with vehicles having different fuel capacities.

\bibliographystyle{IEEEtran}
\bibliography{references}

\end{document}